\documentclass[pdflatex,sn-mathphys-num]{sn-jnl}
\usepackage{graphicx}%
\usepackage{multirow}%
\usepackage{amsmath,amssymb,amsfonts}%
\usepackage{amsthm}%
\usepackage{mathrsfs}%
\usepackage[title]{appendix}%
\usepackage{xcolor}%
\usepackage{textcomp}%
\usepackage{manyfoot}%
\usepackage{booktabs}%
\usepackage{algorithm}%
\usepackage{algorithmicx}%
\usepackage{algpseudocode}%
\usepackage{listings}%

\usepackage{dutchcal}
\usepackage{graphicx}
\usepackage{bm}
\usepackage{color}
\usepackage{alphabeta}
\usepackage[english]{babel}
\usepackage{cleveref}
\usepackage{array}
\usepackage{physics}
\usepackage{mathtools}
\usepackage{amsthm}
\usepackage{tensor}
\usepackage{anyfontsize}

\usepackage{dsfont}

\newtheorem{theorem}{Theorem}
\newtheorem*{theorem*}{Theorem}

\newtheorem*{assumption*}{Assumption}

\newcommand{\bc}{\begin{center}}
\newcommand{\ec}{\end{center}}
\def\ba#1{\begin{array}{#1}\displaystyle}
\newcommand{\ea}{\end{array}}

\newcommand{\beq}{\begin{equation}}
\newcommand{\eeq}{\end{equation}}
\newcommand{\beqa}{\begin{eqnarray}}
\newcommand{\eeqa}{\end{eqnarray}}
\newcommand{\no}{\nonumber}
\newcommand{\n}{\nonumber\\}
\newcommand{\bi}{\begin{itemize}}
\newcommand{\ei}{\end{itemize}}

\def\t#1{\tilde{#1}}

\def\b#1{\bar{#1}}
\def\frc#1#2{\frac{#1}{#2}}

\newcommand{\p}{\partial}

\newcommand{\Z}{{\mathbb{Z}}}
\newcommand{\N}{{\mathbb{N}}}
\newcommand{\R}{{\mathbb{R}}}
\newcommand{\C}{{\mathbb{C}}}

\newcommand{\ep}{\epsilon}

\newcommand{\1}{{\bf 1}}

\DeclareMathOperator{\dist}{{\rm dist}}

\DeclareMathOperator{\supp}{{\rm supp}}

\newcommand{\halmos}{\rule{1ex}{1.4ex}}
\newcommand{\eproof}{\hspace*{\fill}\mbox{$\halmos$}}
\begin{document}

\title[Rigorous bound on hydrodynamic diffusion for chaotic open spin chains]{Rigorous bound on hydrodynamic diffusion for chaotic open spin chains}


\author*[1]{\fnm{Dimitrios} \sur{Ampelogiannis}}\email{dimitrios.ampelogiannis@kcl.ac.uk}

\author[1]{\fnm{Benjamin} \sur{Doyon}}

\affil[1]{\orgdiv{Department of Mathematics}, \orgname{King's College London}, \orgaddress{\street{Strand}, \city{London}, \postcode{WC2R 2LS}, \country{UK}}}


\abstract{The emergence of diffusion is one of the deepest physical phenomena observed in many-body interacting, chaotic systems. But establishing rigorously that correlation functions, say of the spin, expand diffusively, remains one of the most important problems of mathematical physics. We establish for the first time, with Lindbladian evolution, a lower bound on spin diffusion in chaotic, translation-invariant, nearest-neighbor open quantum spin-1/2 chain satisfying a local detailed-balance condition and strong conservation of magnetization. The bound is strictly positive if and only if the local quantum jumps transport spin. Physically, the bound comes from the spreading effects of initial-state macroscopic fluctuations, a mechanism which occurs whenever spin is an interacting ballistic mode. Chaoticity means that the Hilbert space of extensive charges is spanned by magnetization; we expect this to be generic. Our main tool is the Green-Kubo formula, the mathematical technique of projection over quadratically extensive charges, and appropriate correlation decay bounds recently established. Because Lindbladian dynamics is not reversible, the Green-Kubo spin diffusion strength includes a contribution due to irreversibility, which we interpret as encoding the hydrodynamic entropy production that may occur in the forgotten environment. This, we show, vanishes for certain choices of interaction parameters, for which the Lindbladian dynamics becomes reversible. Our methods can be extended to finite or short ranges, higher spins, and other non-Hamiltonian systems such as quantum circuits. As we argue, according to the theory of nonlinear fluctuating hydrodynamics, we further expect these systems to display superdiffusion, and thus have {\em infinite} diffusivity; however this appears to still be beyond the reach of mathematical rigor.}

\keywords{diffusion, open quantum systems, spin chains, hydrodynamics}



\maketitle

\section{Introduction}\label{introduction}

Deriving the laws of hydrodynamics from the microscopic dynamics of many-body systems is a fundamental goal of statistical physics, and doing so rigorously forms an important part of Hilbert's sixth problem. Of great importance is establishing the transport properties, and in particular diffusive transport, in a universal manner, as emerging from the microscopic dynamics. In this endeavor, a crucial step is to determine rigorous lower bounds on the strength of the diffusive spread of response functions, via the celebrated Green-Kubo formula for the Onsager matrix and Einstein's relation connecting it to the diffusion parameters (see e.g.~\cite{spohn_large_1991,DeNardis_Doyon_2019_diffusion_GHD}). The Onsager matrix is manifestly non-negative in parity-time-symmetric dynamics \cite{DeNardis_Doyon_2019_diffusion_GHD}, leading to non-negative entropy production; however showing that it is strictly positive is a difficult problem. In quantum systems, this remains largely unsolved, with very few rigorous results. Focusing on spin chains, one such bound on spin diffusion at infinite temperature, derived by using special quadratically extensive charges, was obtained in \cite{Prosen_2014_diffusion_bound}, and a second bound by the curvature of the Drude weight in \cite{Medenjak_2017_diffusion_bound}. Both rely on the assumed absence of ballistic transport of spin and focus on certain integrable Hamiltonian models, where the charge is explicitly constructed. As far as we are aware, there exists no proof that the diffusion strength is strictly positive in chaotic quantum systems.

In general, there is a mechanism for diffusive effects, coming from large-scale fluctuations of the system's emergent hydrodynamic modes. This happens if hydrodynamic modes exhibit {\em interacting ballistic transport} -- more precisely, the eigenvalues of the flux Jacobian, which are the hydrodynamic velocities generalising the sound velocity in gases, are non-constant functions of the state. Indeed, in such cases, it is expected that fluctuations of the ballistic trajectories induced by initial-state macroscopic, thermodynamic fluctuations \cite{doyon_2023_hydro_longrange,doyon_2023_ballistic_fluctuations}, will lead to spreading that may contribute to diffusion \cite{Medenjak_Yoshimura_2020_Diffusion_bound,hubner_2024_diffusive} and even anomalous effects \cite{gopalakrishnan_2024_universality,krajnik_2022_anomalous_current_fl,yoshimura_2024_anomalous_current,gopalakrishnan_2024_nongaussian_dif}. A proof of this principle, which generalizes the ideas of \cite{Prosen_2014_diffusion_bound}, is found in \cite{doyon_2019_diffusion} where the Onsager matrix is bounded rigorously by a projection formula onto quadratically extensive charges constructed from ballistic modes and representing initial-state fluctuations. This works in one-dimensional systems under the assumption of strong-enough clustering properties of correlations. This principle is important: although chaotic Hamiltonian chains do not admit ballistic transport of spin or energy, as follows from a general result of Kobayashi and Watanabe \cite{kobayashi_2022_vanishing_currents}, as soon as the Hamiltonian structure is broken -- such as in open chains (i.e.~with incoherent processes) or unitary quantum circuits (i.e.~with discrete time) -- interacting ballistic transport may arise. However, the bound in \cite{doyon_2019_diffusion} is not explicitly non-zero, the presence of symmetries can in fact force it to be zero, and it relies on clustering assumptions that were until recently unproven.

In this paper, we construct a family of interacting, translation-invariant, nearest-neighbor spin-$1/2$ open (Lindbladian) infinite-length quantum chains, which conserve total magnetization, and for which we show an explicitly positive bound on the strength of linear-response spin diffusion. The models accounts for spin hopping, both coherently and incoherently, on nearest neighbors, along with phase decoherence effects. For the notion of diffusion to be meaningful in a Lindbladian system, it is natural to impose two properties which suggest that spin transport be of ``hydrodynamic type". First, we require magnetization $M = \sum_{x\in\Z} \sigma_x^3$ to be {\em strongly conserved} (generating a strong symmetry \cite{buca_2012_lindblad_transport}). This implies that we have a local spin conservation law. Second, we require the two-site local incoherent processes to satisfy a {\em local detailed-balance condition} (related to quantum detailed balance \cite{Alicki_1976_detailed_balance,fangola_2007_generator_quantum_markov,carlen_2017_gradient}). We show that this makes Gibbs states with density matrix $e^{\mu M}$ stationary. Thus one may expect emergent hydrodynamics for the local conserved spin thanks to local thermalization. Transport is generically not reversible in Lindbladian systems, but we show that local detailed balance also implies the existence of a (different) time-reverse dynamics that is well defined on local observables in the infinite chain. The family of models we consider is the most general family of translation-invariant, nearest-neighbor spin-1/2 open chains solving these two requirements.

For this family, we show that the spin-spin Onsager coefficient is bounded by projection onto quadratic charges, and that the bound is strictly positive if and only if {\em incoherent transport processes generate a non-zero current}. When this happens, spin is an interacting ballistic mode as, we show, there is a spin current in Gibbs states that depends quadratically on the average local magnetization $\mathsf s$, leading to a hydrodynamic velocity $v(\mathsf s)$ that depends on the $\mathsf s$; the bound is proportional to $(v'(\mathsf s))^2$ (Eq.~\eqref{boundjprime} below). Physically, then, the bound accounts for the spreading effects of initial-state macroscopic fluctuations, the mechanism described above. We establish the bound by applying the result of \cite{doyon_2019_diffusion}, which we make fully rigorous by adapting the proof in order to use recent results on clustering of $n$-th order correlations \cite{ampelogiannis_2024_clustering}. Finally, we show that for certain choices of parameters, the Lindbladian evolution is in fact {\em reversible}, and the irreversibility diffusion strength vanishes. Thus, for this sub-family of models, we have established a strictly positive normal spin diffusion strength. The techniques will apply to any Lindbladian quantum system with strong spin conservation and local detailed balance, and we believe the results will generically hold.

By the theory of nonlinear fluctuating hydrodynamics \cite{spohn_2014_nonlinear}, non-linear currents should give rise to the KPZ equation and superdiffusion. Rigorously establishing that the Onsager coefficient is infinite appears, however, to be yet inaccessible with current techniques.

\subsection{Overview of results: Onsager matrix and explicit bound}\label{section:results}

Because of the lack of time reversal (more precisely, parity-time-reversal) symmetry, constructing a manifestly non-negative quantity that measures diffusion requires special care. We construct a {\em normal diffusion strength} $\mathfrak L_{\rm norm}(t)$ measuring the linear-response diffusive spread of forward spin propagation at time $t$, and an {\em irreversibility diffusion strength} $\mathfrak L_{\rm irr}(t)$ measuring the diffusive spread of the mismatch, due to irreversibility, obtained by performing forward then backward evolution by time $t$:
\begin{align}
    \mathfrak L_{\rm norm}(t) &=
    \frc1{t} \sum_{x\in\Z} (x^2-(v(\mathsf s)t)^2) \langle
    \tau_t^*(s(x)),s(0)\rangle^c\\ 
    \mathfrak L_{\rm irr}(t)&=
    \frc1{2t} \sum_{x\in\Z} x^2\langle \tau_{t}\tau^*_t( s(x)),s(0)\rangle^c.
\end{align}
Here $\langle\cdot,\cdot\rangle^c$ is the connected correlation function in the Gibbs state $e^{\mu M}$ with average magnetization density $\mathsf s = \Tr e^{\mu M}\sigma^3_0/\Tr e^{\mu M}$; $A(x,t) = \tau_t^*(A(x))$ is the Heisenberg-picture forward Lindbladian time evolution of the local operator $A(x)$; $\tau_t$ is the backward time evolution; $v(\mathsf s)$ is the hydrodynamic spin velocity in the Gibbs state; and $s(x)=\sigma_x^3$ is the spin operator at $x$. Note that in the state with zero chemical potential ($\mathsf s=0$), the hydrodynamic velocity vanishes $v(0)=0$. Intuitively, the irreversibility diffusion strength encodes the entropy increase that may occur in the ``forgotten" environment.
Then we show, with full mathematical rigor and without extra assumption, in the general, explicit family of translation-invariant, nearest-neighbor spin-1/2 open chains with strong magnetization conservation and detailed balance, Eqs.~\eqref{time_evolution}, \eqref{hamiltonian_density}, \eqref{jump_operators}, that:
\begin{equation}\label{boundintro}
    \mathfrak L = \liminf_{t\to\infty} \Big(\mathfrak L_{\rm norm}(t) - \mathfrak L_{\rm irr}(t)\Big) \geq
    \frc{(\chi v'(\mathsf s))^2}{8v_{\rm LR}}
\end{equation}
where $\chi$ is the magnetic susceptibility and $v_{\rm LR}$ is the Lieb-Robinson velocity. In the family of models considered, expressions for $\chi$, $v'(\mathsf s)$ and $v_{\rm LR}$ are given in \eqref{chiexplicit}, \eqref{fluxcurvature}, and \eqref{eq:LR_velocity} with \eqref{Vtilde}, respectively. We expect the bound in the form \eqref{boundintro} to be universal for chaotic models with strong spin conservation and local detailed balance. Note that the limits $\lim_{t\to\infty} \mathfrak L_{\rm norm}(t)$ and $\lim_{t\to\infty}\mathfrak L_{\rm irr}(t)$ may in fact exist, but we have not proven this; in this case, the result is for the difference of these limits. The proof is obtained by naturally adapting the standard Green-Kubo formula, for spin diffusion, to Lindbladian evolution:
\begin{align}
   & \mathfrak L = \liminf_{T\to\infty} \frc1{T}\int_0^T\dd t \int_0^T\dd t'\,\sum_{x\in\Z}
    \langle j^-(x,t),j^-(0,t')\rangle^c, \label{Onsager} \\
    &\hspace{6cm}j^-(x,t) = j(x,t) - vs(x,t) .\label{eq:proejcted_spin_current}
\end{align}
We use quadratic charge projections in order to lower-bound this quantity, and we show that it indeed decomposes into the normal and irreversible contributions as above.

The usefulness of the above results relies on the idea that there are no other extensive conserved quantities coupling to the spin current -- as otherwise, additional ballistic spreading is expected to occur and the quantity $\mathfrak L$ is infinite. In its most general form, the Onsager matrix is
\begin{align}
   & \mathfrak L_{mn} = \liminf_{T\to\infty} \frc1{T}\int_0^T\dd t \int_0^T\dd t'\,\sum_{x\in\Z}
    \langle j_m^-(x,t),j_n^-(0,t')\rangle^c, \label{Onsagerintro} \\
    &\hspace{6cm}A^-(x,t) = (1 - \mathbb P) A(x,t)\label{eq:projected_operator}
\end{align}
where $j_m$ are the currents, and $\mathbb P$ is the projector onto the Hilbert space $\mathcal H_1$ of extensive conserved quantities \cite{doyon_hydrodynamic_2022} (we show that this Hilbert space exists for Lindbladian evolution with local detailed balance). As usual, this definition requires the knowledge of all extensive conserved quantities, which is difficult to establish in general, and \eqref{Onsager} with \eqref{eq:proejcted_spin_current} is obtained by assuming that the spin chain is chaotic, in the sense that $\mathcal H_1$ be one-dimensional, spanned by $M$ only (there is only the spin current $j=j_1$). More precisely, we assume that {\em spin transport be chaotic}: the spin current projects only onto $M$. Proving chaoticity is a non-trivial task that is currently out of reach, as it requires one to establish the absence of other extensive charges \cite{Prosen_1999_ergodic,prosen_1998_quantum_invariants,doyon_thermalization_2017,ampelogiannis_long-time_2023, doyon_hydrodynamic_2022}. In Hamiltonian systems there are exciting partial results, showing the absence of charges with local densities (a subset of extensive charges) in certain models \cite{shiraishi_2019_absence_local,chiba_2023_ising_nonintegrability,shiraishi2024absence,park_2024_nonintegrability}. In our understanding it is widely expected that generic interacting spin chains be chaotic.

\section{Set-up} \label{section:set-up}

\subsection{Time evolution}

We consider a quantum spin chain with local spaces $M_2(\mathbb{C})$, spanned by the Pauli matrices and the identity matrix, on each site $x\in\Z$. On each finite $Λ\subset \Z$ we associate the algebra $\mathfrak{U}_{Λ}= \otimes_{i\in Λ} M_2(\C)$. Observables form a quasi-local $C^*$ algebra $\mathfrak{U}$, the norm completion of the algebra $\mathfrak{U}_{\rm loc}= \cup_{Λ } \mathfrak{U}_Λ$ of local observables (supported on finite number of sites). We consider a homogeneous Lindbladian evolution where the generator of the dynamics is defined for each finite $Λ\subset \Z$ and contains both Hamiltonian nearest-neighbor interactions and dissipative two-site terms. The local Hamiltonian, for finite $Λ\subset \Z$, is
\begin{equation}
   H_Λ = \sum_{x\in Λ} h(x), \quad h(x) \in \mathfrak{U}_{\rm loc} .
\end{equation} 
The dissipative part is described by local quantum jump (or Lindblad) operators $L_i(x)\in \mathfrak{U}_{\rm loc}$, $x\in \Z$ for $i$ in some finite index set. On local observables $A\in \mathfrak{U}_{Λ}$ the Lindbladian acts as\footnote{This is the standard form of the Lindbladian in the Heisenberg representation, but written in terms of commutators -- this makes clear that its action is well defined on local operators.}
\begin{equation} \label{time_evolution}
    \mathcal L_{Λ}^*(A) = \partial_t A(t) =  i[H_Λ,A] 
     + \frc12\sum_{x\in Λ,\, i}\big(L_i(x)^\dag [A,L_i(x)] + [L_i(x)^\dag,A]L_i(x)\big). \no
\end{equation}
The thermodynamic limit on the generator $\lim_{Λ \to \Z}\mathcal L_\Lambda^* = \mathcal L^*$ exists on $\mathfrak U_{\rm loc}$; and that of the dynamics
\begin{equation}
    \tau^*_t = \lim_{Λ\to\Z} e^{t\mathcal L^* _\Lambda}
\end{equation}
exists on $\mathfrak U$ and is a strongly continuous completely positive dynamical semigroup  \cite{nachtergaele_2011_LR_irreversible}. 

We require strong spin conservation: the Hamiltonian density, and each quantum jump, are all required to preserve the total spin of the quantum chain,
\begin{equation}\label{strongconserved}
    [M,h(x)]= [M,L_{i}(x)] = 0,\quad M= \sum_{x\in\Z} \sigma^3_x.
\end{equation}
Here and below, for any $A\in \mathfrak{U}_{\rm loc}$ we write $[M,A(x)] = \lim_{\Lambda\to \Z} [M_\Lambda,A(x)]$ with $M_\Lambda= \sum_{x\in\Lambda} \sigma^3_x$. By  spin conservation, the Hamiltonian density and quantum jump operators take the general forms
\begin{equation} \label{hamiltonian_density}
    h(x) = \alpha\sigma^+_x\sigma^-_{x+1}
    + \bar\alpha \sigma^-_x\sigma^+_{x+1}
    + \beta\sigma^3_{x}\mathds{1}_{x+1}
    + \gamma\sigma^3_x\sigma^3_{x+1},
\end{equation}
and
\begin{equation} \label{jump_operators}
    L_i(x) = a_i\sigma^+_x\sigma^-_{x+1}
    + b_i \sigma^-_x\sigma^+_{x+1}
    + c_i\sigma^3_{x}\mathds{1}_{x+1}
    + d_i \mathds{1}_{x}\sigma^3_{x+1}
    + e_i\sigma^3_x\sigma^3_{x+1}
\end{equation}
where $α\in\C$, $\beta,γ\in \R$, $a_i,b_i,c_i,d_i,e_i \in \C$ are constants, which we will refer to as interaction parameters. Here, $σ_x^1$,$σ_x^2$,$σ_x^3 \in \mathfrak{U}_{ \{x \}}$ are the x, y and z Pauli matrices respectively acting on $\mathfrak{U} _{\{x \} }$ and $σ^{\pm}= \frac12(σ^1 \pm i σ^2)$.  Note that for the Hamiltonian density it is not necessary to have the term $\sigma^3_{x+1}$ as under $\sum_{x\in\Z}$ in $H$ this is redundant; while for the Lindblad operator $L_i$ there is no redundancy. Observe that for $|a_i|>|b_i|$ the $i$th process leads to incoherent spin transport towards the left, while for $|a_i|<|b_i|$ it is towards the right; and the sum of terms for $c_i,d_i,e_i$ represent various types of incoherent local dephasing.

Time evolution for Hamiltonian systems satisfies the Lieb-Robinson bound \cite{Lieb:1972wy}, which can be extended to Lindbladian dynamics \cite{poulin_2010_LR_markovian,nachtergaele_2011_LR_irreversible}, that yields a light-cone effect 
\begin{equation}
    \norm{[A(x,t),B] }\leq C_{A,B} e^{-λ(x-υ_{\rm LR}|t|)}.
\end{equation}
The Lieb-Robinson velocity $υ_{\rm LR}$ for \eqref{time_evolution} can be estimated by the results \cite{nachtergaele_2011_LR_irreversible,sims_2011_LR}:
\begin{equation} \label{eq:LR_velocity}
    υ_{\rm LR} = 4eζ(2) V , \quad V=2 \norm{h(x)} + 2 \sum_i \norm{L_i(x)}
\end{equation}
where $ζ$ is the Riemann function, and we can estimate for \eqref{hamiltonian_density}, \eqref{jump_operators}:
\begin{equation}\label{Vtilde}
    V \leq \t V = 2( 2|α|+|β|+|γ|) + 2 \sum_i( |a_i|+|b_i|+|c_i|+|d_i|+|e_i|).
\end{equation}
Lieb-Robinson bounds are generally loose, and improvements can be made in the velocity estimates \cite{Wang_2020_tightening_LR}.

\subsection{Gibbs states and inner products}

A state $ω$ is defined as a positive, normalised to 1, linear functional of $\mathfrak{U}$ and $ω(A)$ is interpreted as the ensemble average of the observable $A$. We consider the system to be in a Gibbs state $\omega_\mu$, for the total spin $M$, with chemical potential $\mu$, defined as 
\begin{equation} \label{gibbs_state}
    \omega_\mu(A) = \lim_{Λ \to \Z} \frc{\Tr_{\prod_{i\in\Lambda} \C^2}( e^{\mu M_Λ}A)}{\Tr_{\prod_{i\in\Lambda} \C^2}(e^{\mu M_Λ})}.
\end{equation}
Below, when there is no ambituity possible, we will write $\Tr(\rho A)/\Tr(\rho)$ to mean the limit $\Lambda\to\Z$ of $\Tr_{\prod_{i\in\Lambda} \C^2}(\rho_\Lambda A)/\Tr_{\prod_{i\in\Lambda} \C^2}(\rho_\Lambda)$. It is clear that the state $\omega_\mu$ satisfies a strong clustering condition: $\omega_\mu(AB) = \omega_\mu(A)\omega_\mu(B)$ for every $A\in\mathfrak U_\Lambda$, $B\in\mathfrak U_{\Lambda'}$ with $\Lambda\cap \Lambda'=\emptyset$.

Observables in $\mathfrak{U}$ can be organised in  equivalence classes to construct Hilbert spaces with inner products which encode the relevant contributions at different hydrodynamic scales, see \cite{doyon_2019_diffusion,doyon_hydrodynamic_2022, ampelogiannis_long-time_2023}. It is on these Hilbert spaces, and not on $\mathfrak U$, that certain important properties hold. From the state $ω$ we define the sesquilinear forms $\langle \cdot , \cdot \rangle_0$ and $\langle \cdot , \cdot \rangle_1$  as
\begin{align}\label{innerproduct}
    \langle A,B\rangle_0   &\coloneqq \langle A^\dag,B\rangle^c \coloneqq ω( A^{\dagger} B) - ω(A^{\dagger})ω(B)\\
    \langle A,B\rangle_1   &\coloneqq \sum_{x \in \Z} \langle A(x),B\rangle_0, \quad A,B\in \mathfrak{U}_{\rm loc}.
\end{align}
We define the equivalence relations $A \sim_k A^{\prime}$ iff $\langle A- A^{\prime}, A- A^{\prime} \rangle_k=0$, $k=0,1$, and form the quotient spaces $\mathcal{V}_k=\mathfrak{U}_{\rm loc}/ \sim_k$, 
which are pre-Hilbert spaces. Completing these spaces yields Hilbert spaces $\overline{\mathcal{V}_k} =\mathcal{H}_k$ where $\langle \cdot , \cdot \rangle_k$ is an inner product. $\mathcal H_0$ is simply related to the Gelfand-Naimark-Segal Hilbert space \cite{bratteli_operator_1987}. $\mathcal H_1$ is a Hilbert space whose elements we can think of as extensive quantities: note that any local observable is equivalent to its space translates. Hence we can think of the equivalence class of any local $A$ as the formal sum $\sum A(x)$ (the total $A$ of the spin chain) and the inner-product carries information about the convergent connected correlations of such observables. 

Time evolution generated by $\mathcal L^*$ gives rise to a time evolution operator on (the dense subspace of local elements of) $\mathcal H_k$. For $\mathcal H_1$, this requires the Lieb-Robinson bound, adapting the proof of \cite[Thm 5.11]{doyon_hydrodynamic_2022} and using the fact that $τ_t^*$ is a contraction of $\mathfrak{U}$: $\norm{τ_t^*A} \leq \norm{A}$, \cite{nachtergaele_2011_LR_irreversible}. This is done by the local approximation $σ_n(τ^*_tA)$ of time evolved observables $τ^*_tA$, $A\in \mathfrak{U}_{\rm loc}$\cite{nachtergaele_quasi-locality_2019}: $\norm{σ_n(τ^*_tA)-τ^*_tA} \leq C_A e^{-λ(n- υ_{LR}|t|)}$, which can be used to show that  $ \lim_n \langle σ_n(τ^*_tA), B \rangle_1 = \langle τ^*_tA,B \rangle_1$ exists for all $A,B \in \mathfrak{U}_{\rm loc}$ and conclude that $σ_n(τ^*_tA)$ converges weakly in $\mathcal{H}_1$. This defines $τ^{\mathcal{H}_1}_t: \mathcal{V}_1 \to \mathcal{H}_1$ by $τ^{\mathcal{H}_1}_t(A) = \rm w-\lim σ_n(τ^*_tA)$.
Below, by abuse of notation we will use the same symbols, $\tau_t$ and $\tau_t^*$, for time evolution on $\mathfrak U$, and on the Hilbert spaces $\mathcal H_k$, when there is no confusion possible.

\subsection{Local detailed balance}

The set of quantum jump operators is required to satisfy the {\em local detailed-balance condition}
\begin{equation} \label{detailed_balance}
    \sum_i [L_{i}(x),L_{i}(x)^\dagger] = o(x+1)-o(x)
\end{equation}
for some ``detailed-balance current" operator $o(x)\in \mathfrak U_{\rm loc}$. Note that by strong spin conservation and locality of $o(x)$, we have $[M,o(x)]=0$ for all $x$. The local detailed-balance condition, along with strong spin conservation, guarantees that the Gibbs state $\omega_\mu$ is stationary. Indeed, for any $A\in\mathfrak U_{\rm loc}$,
\begin{align}
    \omega_\mu(\mathcal L^*(A))
    &= \sum_{x\in\Z}\sum_{i}\omega_\mu(A[L_i(x),L_i(x)^\dagger]) \\
    &= \sum_{x\in\Z} \big(\omega_\mu (A o(x+1))-\omega_\mu(Ao(x))\big) = 0 \no
\end{align}
where in the first line we used the definition of $\mathcal L^*$ and strong spin conservation, and in the last we used telescopic summation, and clustering and homogeneity of the state $\omega_\mu$. By a straightforward calculation we find that the local detailed-balance condition is equivalent to
\begin{align}\label{condition}
    \sum_i a_i(\b d_i - \b c_i) = \sum_i \b b_i(d_i-c_i),
\end{align}
with $o(x) = \sum_i\frac{|b_i|^2-|a_i|^2}2 \sigma^3_x$.

Local detailed balance has two important consequences:
\begin{theorem}
The local Lindbladian defined by exchanging $L_i(x)\leftrightarrow L_i(x)^\dagger$ and changing the sign of the Hamiltonian,
\begin{equation}
    \label{time_evolution_conjugate}
    \mathcal L(A) = -i[H,A] 
     + \frc12\sum_{x\in \Z,\, i}\big(L_i(x) [A,L_i^\dag(x)] + [L_i(x),A]L_i^\dag(x)\big)
\end{equation}
is conjugate to $\mathcal L^*$ with respect to both inner products above,
\begin{equation}\label{adjointL}
    \langle \mathcal L^*(A),B\rangle_k =
    \langle A,\mathcal L(B)\rangle_k,\quad
    \langle \mathcal L(A),B\rangle_k =
    \langle A,\mathcal L^*(B)\rangle_k,
    \quad
    A,B\in\mathfrak U_{\rm loc}
\end{equation}
and gives rise to a strongly continuous completely positive dynamical semigroup $\tau_t=\lim_{\Lambda\to\Z} e^{t\mathcal L_\Lambda}$ on $\mathfrak U$ with a Lieb-Robinson bound with \eqref{eq:LR_velocity}. (2) Both $\tau_t^*$ and $\tau_t$ are contracting with respect to both inner products,
\begin{equation}\label{contracting}
    \langle \tau_t^*(A),\tau_t^*(A)\rangle_k \leq 
    \langle A,A\rangle_k,\quad
    \langle \tau_t(A),\tau_t(A)\rangle_k \leq 
    \langle A,A\rangle_k, \quad A\in\mathfrak U_{\rm loc},
\end{equation}
and thus can be extended to bounded operators forming semigroups on $\mathcal H_k$.
\end{theorem}
Explicitly, $\mathcal L$ is obtained from $\mathcal L^*$ by exchanging $b_i\leftrightarrow \b a_i$ and $c_i,d_i,e_i\leftrightarrow \b c_i, \b d_i,\b e_i$ in \eqref{jump_operators}. We believe the semigroups formed by $\tau_t$ and $\tau^*_t$ on $\mathcal H_k$ can be shown to be strongly continuous (but this does not play any role here). As a consequence that $\tau_t^*$ is a bounded operator, the space of extensive conserved charges $\mathcal Q = \{q\in \mathcal H_1:\tau_t^*(q) = q\}$ \cite{doyon_hydrodynamic_2022}, is a closed subspace, and the projection $\mathbb P:\mathcal H_1\to \mathcal Q$ in \eqref{Onsagerintro} makes sense. We interpret $\tau_t$ as the {\em backward time evolution} of the system; there is also a notion of extensive conserved quantities for backward time evolution. Note that $\mathcal L$ {\em does not} take the form the Lindbladian usually takes on density matrices; by contrast, the form above indeed makes sense on local operators in the infinite-volume setting. We show in Appendix \ref{apprev} with certain choices of parameters, $\mathcal L = -\mathcal L^*$, in which case $\tau_t = \tau_{-t}^*$ and time evolution is unitary on $\mathcal H_k$: the dynamics is reversible.

Quantum detailed balance conditions have been discussed \cite{Alicki_1976_detailed_balance,fangola_2007_generator_quantum_markov,carlen_2017_gradient}; Eq.~\eqref{adjointL} is a related formulation in the infinite-volume context.
\proof 
The form of $\mathcal L$ is the same as that of $\mathcal L^*$, exchanging $L_i(x)\leftrightarrow L_i(x)^\dag$. Therefore, by the above discussion, it exists on $\mathfrak U_{\rm loc}$, generates a strongly continuous dynamical semi-group on $\mathfrak U$ with a Lieb-Robinson bound on $\mathfrak U_{\rm loc}$ and with the velocity \eqref{eq:LR_velocity}, and this time evolution gives rise to a well defined evolution operator $\tau_t$ on the local elements of $\mathcal H_k$, $k=0,1$, and on their time translates by both $\tau_t$ and $\tau_t^*$.

What remains to be shown are Eqs.~\eqref{detailed_balance} and \eqref{contracting}. Eq.~\eqref{detailed_balance} follows from evaluating connected correlation functions: using strong spin conservation \eqref{strongconserved} we have
\begin{equation}
    \langle \mathcal L^*(A)^\dag,B\rangle^c
    =
    \langle A^\dag,\mathcal L(B) \rangle^c
    +
    \sum_{x\in\Z,\,i}\langle A^\dag,\{B,[L_i(x),L_i(x)^\dag]\}\rangle^c
\end{equation}
and using \eqref{detailed_balance} the series vanishes by telescopic summation and clustering. The opposite order of $\mathcal L,\mathcal L^*$ is obtained similarly.

For Eq.~\eqref{contracting}, this follows because local detailed balance essentially implies that the Lindbladian is unital, for which contraction properties are known \cite{perez-garcia_2006_contractivity}. However we are in the infinite-volume setting, so care must be taken. We provide a direct proof. We consider $\tau_t^*$, as the same proof holds for $\tau_t$.

It is sufficient to have the proof for $\langle \cdot,\cdot\rangle_0$. Indeed, by translation invariance and the existence of the limit defining $\langle A,A\rangle_1$,
\begin{equation}
    \langle A,A\rangle_1
    = \lim_{\Lambda\to\Z}
    \frc1{|\Lambda|}
    \langle \sum_{x\in\Lambda}A(x),\sum_{x'\in\Lambda}A(x')\rangle_0.
\end{equation}
We consider, for any $A\in\mathfrak U_{\rm loc}$,
\begin{equation}\label{proofbounded1}
    \frc{d}{dt} \langle e^{t\mathcal L^*_\Lambda} (A),e^{t\mathcal L^*_\Lambda} (A)\rangle_0 =
    \omega_\mu(\mathcal L_\Lambda^*(B^\dag)B)
    + \omega_\mu(B^\dag \mathcal L_\Lambda^*(B))
\end{equation}
with $B = e^{t\mathcal L^*_\Lambda} (A)$, and
where differentiability holds because for $\Lambda$ a finite set, we have finite matrices. The right-hand side is, by using strong magnetization conservation and the cyclic property of the trace
\begin{equation}
    -\sum_{x,i} \omega_\mu \Big(L_i(x) B^\dag BL_i(x)^\dag 
    + B^\dag L_i(x)^\dag L_i(x)B
    - L_i(x)^\dag B^\dag L_i(x) B
    - L_i(x) B^\dag L_i(x)^\dag B \Big)
\end{equation}
where the sum on $x$ is over $\Lambda = [-N,N]\cap \Z$. Now bound (using Cauchy-Schwartz inequality)
\begin{equation}
    |\omega_\mu(L_i(x)^\dag B^\dag L_i(x)B)| \leq \sqrt{\omega_\mu(L_i(x)^\dag B^\dag B L_i(x))\omega_\mu(B^\dag L_i(x)^\dag  L_i(x)B)}
\end{equation}
and similarly,
\begin{equation}
    |\omega_\mu(L_i(x) B^\dag L_i(x)^\dag B)| \leq \sqrt{\omega_\mu(L_i(x) B^\dag B L_i(x)^\dag)\omega_\mu(B^\dag L_i(x)  L_i(x)^\dag B)}.
\end{equation}
Writing $\omega_\mu (L_i(x)^\dag L_i(x) B^\dag B) = \frc12 \omega_\mu (L_i(x)L_i(x)^\dag B^\dag B) + \frc12 \omega_\mu (L_i(x)^\dag L_i(x)B^\dag B) + \frc12 \omega_\mu ([L_i(x)^\dag ,L_i(x)]B^\dag B) $, and using local detailed balance and translation-invariance of the state, we have
\begin{equation}
    \sum_{x,i} \omega_\mu (L_i(x) B^\dag BL_i(x)^\dag)
    =
    \frc12 \sum_{x,i} \omega_\mu \Big(L_i(x) B^\dag BL_i(x)^\dag+
    L_i(x)^\dag B^\dag BL_i(x)\Big) + \epsilon
\end{equation}
where $\epsilon = \epsilon(t) = \frc12\langle o(N+1) - o(-N),B^\dag B\rangle^c$ (note that $B$ depends on $t$).
Similarly,
\begin{equation}
    \sum_{x,i} \omega_\mu (B^\dag L_i(x)^\dag L_i(x) B)
    =
    \frc12 \sum_{x,i} \omega_\mu \Big(B^\dag L_i(x) L_i(x)^\dag B^\dag+
    B^\dag L_i(x)^\dag  L_i(x) B \Big) + \b\epsilon
\end{equation}
where $\b\epsilon = \b\epsilon(t) = \frc12\langle o(N+1) - o(-N),B B^\dag\rangle^c$.
With
\begin{align}
    r_1&=\sqrt{\omega_\mu(L_i(x)^\dag B^\dag B L_i(x))}\\
    r_2&=\sqrt{\omega_\mu(L_i(x)B^\dag B L_i(x)^\dag)}\\
    r_3 &= \sqrt{\omega_\mu( B^\dag L_i(x)^\dag L_i(x) B )}\\
    r_4 &= \sqrt{\omega_\mu( B^\dag L_i(x)L_i(x)^\dag B )}
\end{align}
(keeping $x,i$ implicit), the absolute value of the right-hand side of \eqref{proofbounded1} is
\begin{equation}
    \leq - \frc12 \sum_{x,i}
    \Big(
    r_1^2 + r_2^2 + r_3^2 + r_4^2
    - 2r_1r_3 - 2r_2r_4
    \Big)+|\ep+\b\ep|
    =
    - \frc12 \sum_{x,i}
    \Big(
    (r_1-r_3)^2 + (r_2-r_4)^2
    \Big)+|\ep+\b\ep|.
\end{equation}
Therefore, we have
\begin{equation}
    \Big|\frc{d}{dt} \langle e^{t\mathcal L^*_\Lambda} (A),e^{t\mathcal L^*_\Lambda} (A)\rangle_0\Big| \leq |\ep+\b\ep|.
\end{equation}
By the Lieb-Robinson bound, fore every $t>0$ there is $\ep_N>0$ such that $|\ep(t')+\b\ep(t')|\leq \ep_N$ uniformly on $t'\in[0,t]$, and such that $\lim_{N\to\infty} \ep_N = 0$. Hence, integrating from $t=0$, we obtain
\begin{equation}
    \Big|\langle e^{t\mathcal L^*_\Lambda} (A),e^{t\mathcal L^*_\Lambda} (A)\rangle_0 - \langle A,A\rangle_0\Big| \leq  \ep_N t.
\end{equation}
Taking the limit $\Lambda\to\Z$ (implying $N\to\infty$ and therefore $\ep_N\to0$), the result follows.
\eproof

\section{Lindbladian Onsager coefficient} 

\subsection{Extensive conserved quantities and chaoticity} In general, the relevant degrees of freedom for the hydrodynamic scale are the extensive conserved quantities, formally $Q_n = \sum_{x \in \Z}q_n(x)$, afforded by the microscopic dynamics, $\mathcal{L}^*(Q_n)=0$. With $q_n(x,t)$, $j_n(x,t)$ conserved densities and their currents, we have conservation laws $\partial_t q_n(x,t) + j_n(x,t)- j_n(x-1,t)=0$; it is those that control hydrodynamics. In fact, $Q_n$ as written above is formal, and does not exist in $\mathfrak U$. The correct construction of the space of extensive conserved quantities, which includes those with local or quasi-local densities ($q_n(x)\in\mathfrak U_{\rm loc}$ or $\mathfrak U$) and possibly more, and the existence of the local currents, are studied in \cite{doyon_hydrodynamic_2022}. This space, $\mathcal Q$, is in fact a Hilbert space, a closed subspace of $\mathcal H_1$ (because $\tau_t^*$ is bounded on $\mathcal H_1$), thus it is complete in the topology induced by $\mathcal H_1$. It is shown in \cite{doyon_hydrodynamic_2022} that, for Hamiltonian systems, it is this full Hilbert space $\mathcal Q$ that contributes to the Euler scale; it determines the exact Drude weight via a projection mechanism. Estimating the diffusive strength of any hydrodynamic mode requires, as per the Green-Kubo formula, subtracting its projection onto ballistic modes, and thus requires the full knowledge of this Hilbert space.

In Lindbladian systems, the closed space of extensive conserved quantities is likely to still play a crucial role, but the question of how is beyond the scope of this paper. Still, in the Green-Kubo formula \eqref{Onsagerintro} it is natural to project out time-invariant quantities. Systems with Lindbladian \eqref{time_evolution} may admit an infinity of conserved quantities, e.g.~if integrable, which may happen also in the presence of incoherent processes \cite{ziolkowska_2020_yang_baxter,essler_2020_integrablility_lindbladian}; or a finite number, in which case we may characterise the system as being {\em chaotic}. In many integrable models, one knows how to write down the conserved charges iteratively, but proving that this set is complete, and no other charges exist, remains a conjecture.
Likewise, proving that a system is chaotic remains an open problem. In our understanding, it is however expected that non-chaotic models are of measure $0$ in the space of interactions.

We therefore assume that the only conserved extensive quantity of \eqref{time_evolution} is the magnetization $Q_1 = M= \sum_{x \in \Z} σ^3_x$; that is $\mathcal Q = {\rm span}(M)$. More precisely, in fact, we only need to assume that the spin current $j$ (as an element of $\mathcal H_1$) projects onto $M$; physically, we assume that spin transport be chaotic. We emphasize that this assumption is only used in addressing the projection operator in the Green-Kubo formula \eqref{Onsagerintro}. Once the formula is written for the projection onto magnetization, our lower bound result is rigorous. Note that $M$ is also conserved by the backward time evolution $\tau_t$.

\subsection{Onsager coefficient and its decomposition}

The spin current $j(x)$ for spin density $s(x) = \sigma_x^3$ is defined by
\begin{equation} \label{eq:spin_conservation}
\mathcal L^*(s(x)) + j(x)-j(x-1) = 0
\end{equation}
and decomposes into a Hamiltonian (coherent) and Lindbladian (incoherent) parts, $j(x) = j^{\rm H}(x) + j^{\rm L}(x)$. Both can be straightforwardly evaluated, but we will only need the incoherent part (here we have used \eqref{condition} to slightly simplify the expression):
\begin{align}
    j^{\rm L}(x) &= \sum_i\Big(
    2|b_i|^2 P_{x}^+P_{x+1}^- - 2|a_i|^2 P_x^-P_{x+1}^+ \n & \quad + \big(2a_i(\b d_i-\b c_i) + \b e_i a_i - e_i \b b_i\big)\sigma_x^+\sigma_{x+1}^-
    \n & \quad + \big(2\b a_i(d_i-c_i) + e_i \b a_i - \b e_i b_i\big)\sigma_x^-\sigma_{x+1}^+
    \Big) \label{Lcurrent}
\end{align}
where $P^\pm_x = (1\pm\sigma_x^3)/2$ are local projections onto positive (up) and negative (down) spins. Our goal is to estimate the diffusive transport coefficient \eqref{Onsager} for $j(x)$.

With the chaoticity assumption, formula \eqref{Onsagerintro} simplifies to ($\mathfrak L = \mathfrak L_{11}$):
\begin{align}
   & \mathfrak L = \liminf_{T\to\infty} \frc1{T}\int_0^T\dd t \int_0^T\dd t'\,\sum_{x\in\Z}
    \langle j^-(x,t),j^-(0,t')\rangle^c, \label{Onsager1} \\
    &\hspace{5cm} j^-(x,t) = j(x,t) - s(x,t) \chi^{-1} \langle M,j(0,0)\rangle^c \label{eq:proejcted_spin_current_chi}
\end{align}
where
\begin{equation}
    \chi = \langle M,s(0,0)\rangle^{\rm c}
\end{equation}
is the magnetic susceptibility.

Note that $\sum_{x\in\Z} \langle j^-(x,t),j^-(0,t')\rangle^c = \langle j^-(t),j^-(t')\rangle_1$ is the $\mathcal H_1$ inner product, that $\langle M,a(0,0)\rangle^{\rm c} = \langle s,a\rangle_1$, so that $s(x,t) \chi^{-1} \langle M,j(0,0)\rangle^c = (\mathbb Pj)(x,t)$ is indeed the projection within the Hilbert space $\mathcal H_1$. The quantity of which we take the limit in \eqref{Onsager1} exists and is finite, by the Lieb-Robinson bound. Further, expression \eqref{Onsager1} is manifestly non-negative, $\mathfrak L\geq 0$, as, by averaging over positions using translation invariance, it is the expectation value of the square of a hermitian operator (see e.g.~\cite{doyon_thermalization_2017}). We take the infimum limit $\liminf_{T\to\infty}$ (\cite[Eqs 39, 58]{doyon_2019_diffusion}), as the limit itself, in \eqref{Onsager1}, might not exist. The result may also diverge to infinity (indicating superdiffusion). By contrast to the Hamiltonian dynamics, because Lindbladian time evolution is not an automorphism of the operator algebra, $A(t) B(t) \neq \tau^*_t(AB)$, even though the state is stationary, $\omega (\tau_t^*(A)) = \omega(A)$, the quantity $\langle j^-(x,t),j^-(0,t^{\prime})\rangle^c$ is not a function of $t-t^{\prime}$.

What is the interpretation of formula \eqref{Onsager1}?

Technically diffusion occurs within a derivative expansion of the current, accounting for the principle that the system tends towards entropy maximization. In a long-wavelength state $\langle\cdots\rangle$, one writes phenomenologically $\langle j(x,t)\rangle = \mathsf j(\langle s(x,t)\rangle) -\frac12 \mathfrak D(\langle s(x,t)\rangle)\partial_x \langle s(x,t)\rangle + \ldots$ where $\mathsf j(\mathsf s) = \omega_\mu(j)$ for $
\mu$ such that $\mathsf s = \omega_\mu(s)$. It is a simple matter to evaluate the current and velocity: an explicit calculation of the trace is possible as it factorises into the individual sites, and we obtain
\begin{equation}
    \mathsf j = \omega_\mu(j^{\rm L}(x)) = \sum_i \frc{|b_i|^2-|a_i|^2}{2\cosh^2\mu}.
\end{equation}
As $\mathsf s = \omega_\mu(\sigma_x^3) = \tanh\mu$, we can express the current as a function of the spin density, obtaining
\begin{equation}\label{js}
    \mathsf j(\mathsf s) = 
    \sum_i(|b_i|^2-|a_i|^2)
    \frc{1-\mathsf s^2}2.
\end{equation}
Linear response arguments for the (phenomenologically justified) diffusive-scale hydrodynamic equation $\p_t \mathsf s(x,t) + v(\mathsf s(x,t)) \p_x \mathsf s(x,t)  = \frc12\p_x (\mathfrak D(\mathsf s(x,t)) \p_x \mathsf s(x,t))$ then give an equation for the correlation $\langle s(x,t),s(0,0)\rangle^c$, solved as a Gaussian centered around the velocity $v(\mathsf s)= d \mathsf j(\mathsf s)/d \mathsf s$:
\begin{equation}\label{diffusiveform}
    \langle s(x,t),s(0,0)\rangle^c \sim \frc{χ}{\sqrt{2\pi \mathfrak D_{\rm norm} t}}\, e^{-(x-vt)^2/(2\mathfrak D_{\rm norm}t)}\qquad \mbox{(hydrodynamic linear response)}
\end{equation}
where we recall that $A(x,t) = \tau^*_t(A(x))$ is the Heisenberg-picture time evolution of the local observable $A(x)\in \mathfrak{U}_{\rm loc}$. Here we have renamed $\mathfrak D\to\mathfrak D_{\rm norm}$ in order to distinguish it from the ``irreversibility diffusion" discussed below; this is {\em normal} diffusion.

Recall that we have different forward and backward time evolutions. According to hydrodynamic intuition, both should lead to diffusion, thus there should be diffusion associated to irreversibility as well. Time evolution $\tau_t$ leads to the opposite average current $-\mathsf j(\mathsf s)$, Eq.~\eqref{js}, thus the opposite hydrodynamic velocity. Evolving forward, then backward in time, and phenomenologically assuming diffusive extension around the end-point of the resulting ballistic trajectory, we thus expect, in analogy to hydrodynamic linear response,
\begin{equation}\label{diffusiveirr}
    \langle \tau_t \tau_t^*s(x),s(0,0)\rangle^c \sim \frc{χ}{\sqrt{4\pi \mathfrak D_{\rm irr} t}}\, e^{-x^2/(4\mathfrak D_{\rm irr}t)}\qquad \mbox{(hydrodynamic linear response).}
\end{equation}
This is {\em irreversibility} diffusion.

Using the relation
\begin{equation}
    v(\mathsf s) = \frc{d \mathsf j}{d\mu} \Big/ \frc{d\mathsf s}{d\mu} = \chi^{-1}\langle M,j(0,0)\rangle^c
\end{equation}
we write from \eqref{eq:proejcted_spin_current_chi}
\begin{equation}
    j^-(x,t) = j(x,t) - v s(x,t)
\end{equation}
and \eqref{Onsager} (or \eqref{Onsager1}) with \eqref{eq:proejcted_spin_current} follows. Using conservation laws, we obtain the physical meaning of \eqref{Onsager}. Following the steps of \cite[App A]{DeNardis_Doyon_2019_diffusion_GHD} (the steps can be made fully rigorous, using discrete-space integration by parts $\sum_x x \p_x j(x) = -\sum_x j(x)$, $\sum_x x^2 \p_x j(x) = -\sum_x (1+2x)j(x)$, with $\p_x j(x) = j(x)-j(x-1)$, and the Lieb-Robinson bound and clustering of correlation functions for boundary terms to vanish), we find that it decomposes into (the infimum limit of) the difference
\begin{equation}
    \mathfrak L := \liminf_{t\to\infty} (\mathfrak L_{\rm norm}(t)-\mathfrak L_{\rm irr}(t))
\end{equation}
of a normal linear-response diffusion strength
\begin{align}
    \mathfrak L_{\rm norm}(t) &=
    \frc1{t} \sum_{x\in\Z} (x^2-(vt)^2) \langle
    s(x,t),s(0,0)\rangle^c\\ &= 
    \frc1t \sum_{x\in\Z} (x^2 - (vt)^2) \frc{d}{d\mu}\frc{\Tr \Big(e^{M + \mu s(0)} s(x,t)\Big)}{\Tr e^{M + \mu s(0)} }\Bigg|_{\mu=0},
\end{align}
and an irreversibility diffusion strength
\begin{align}
    \mathfrak L_{\rm irr}(t)&=
    \frc1{2t} \sum_{x\in\Z} x^2\langle \tau_{t}\tau^*_t( s(x)),s(0,0)\rangle^c
    \\ &= 
    \frc1{2t} \sum_{x\in\Z} x^2\frc{d}{d\mu}\frc{\Tr \Big(e^{M + \mu s(0)} \tau_{t}\tau^*_t(s(x))\Big)}{\Tr e^{M + \mu s(0)} }\Bigg|_{\mu=0}
\end{align}
the linear response of the spin evolved forward, and then backward, for the same time. If the correlation functions satisfy hydrodynamic diffusive equations, then \eqref{diffusiveform} and \eqref{diffusiveirr} hold and we have Einstein relations for the limits $t\to\infty$:
\begin{equation}
    \mathfrak L_{\rm norm}(\infty) = \mathfrak D_{\rm norm}\chi,\quad \mathfrak L_{\rm irr}(\infty) = \mathfrak D_{\rm irr}\chi\quad
    \mbox{(hydrodynamic linear response)},
\end{equation}
which provides the physical meaning of $\mathfrak L$. With the choices of parameters making the dynamics reversible, Appendix \ref{apprev}, we have $\tau_t = \tau_{-t}^*$ and therefore $\mathfrak L_{\rm irr}=0$.

\section{Main result: strictly positive spin diffusion} 
To prove that the system exhibits at least diffusive transport, it suffices to find a lower, strictly positive, bound for $\mathfrak{L}$. In \cite{doyon_2019_diffusion}  a general bound is obtained in one-dimensional systems, under the assumption of clustering of $n$-th order connected correlations under space-time translations, which until recently remained without a proof. In finite range quantum spin chains the Gibbs states for nonzero temperatures are exponentially clustering \cite{Araki:1969bj}, and the state $\omega_\mu$ is manifestly clustering because of ultra-locality of $M$. Combined with the Lieb-Robinson bound this implies space-like clustering for two-point connected correlations, \cite[Appendix C]{ampelogiannis_long-time_2023}. Recently in \cite{ampelogiannis_2024_clustering} the authors established  clustering of $n-$th order cumulants from two-point clustering, which is slightly different from the assumed clustering in \cite{doyon_2019_diffusion}. In particular, applying \cite[Theorem V.4]{ampelogiannis_2024_clustering} in the case of exponential clustering, we find that the $n$-th order connected correlations of $ω_μ$ satisfy:
\begin{equation} \label{eq:clustering}
\arraycolsep=1.4pt\def\arraystretch{1.5}
\begin{array}{*3{>{\displaystyle}l}p{5cm}}
    \langle  A_1(x_1,t_1), A_2(x_2,t_2), \ldots ,A_n(x_n,t_n) \rangle^c \leq C_{A_1,\ldots,A_n}  
    e^{ -\nu z}, \\
     \text{where } z=\max_i \min_j \{ \dist(A_i(x_i),A_j(x_j))\}
    \end{array}
\end{equation}
for any $n\in \N$, $A_1,\ldots,A_n \in \mathfrak{U_{\rm loc}}$, $x_1,\ldots,x_n\in \Z$ and $t_1,\ldots ,t_n \in [-υ^{-1}z+1, υ^{-1}z-1]$, where $C_{A_1,\ldots,A_n},\nu>0 $ constants \footnote{The constant $C_{A_1,\ldots,A_n}$  depends linearly on the norms of the observables and polynomially on the sizes of their supports}. Having established \Cref{eq:clustering}, we can slightly modify the proof of \cite[Theorem 5.1]{doyon_2019_diffusion}, see \Cref{appendix4}, to rigorously obtain the bound\footnote{Note that there is a difference of a factor of $2$ with \cite[Theorem 5.1]{doyon_2019_diffusion}, this is because in the present set-up the Onsager coefficient \eqref{Onsager1} is defined as an integral over $t\geq 0$ instead of  $\R$}:
\begin{equation} \label{bound}
    \mathfrak{L} \geq \mathfrak L_{\rm lower} := \frac{|\langle M ,M,j^-\rangle^c|^2}{8υ_{\rm LR} (\langle M,s \rangle^c)^2}.
\end{equation}
Here $M=\sum_x{σ_x^3}$ is the magnetisation, $s=σ^3_0$ the spin density\footnote{The result is irrespective of the site $x\in \Z$ of the density}, $υ_{LR}$ the Lieb-Robinson velocity \eqref{eq:LR_velocity} and $j^-$ is the projection of the local spin currect \eqref{eq:proejcted_spin_current}. If the three point connected correlation $\langle M,M,j^-\rangle^c$ is non-zero, then the Onsager coefficient is strictly positive. We show the following:

\begin{theorem} \label{theorem:1}
    Consider a quantum spin chain with Lindbladian evolution \eqref{time_evolution} with nearest neighbor Hamiltonian interaction \eqref{hamiltonian_density} and Lindblad operators \eqref{jump_operators}, the spin current \eqref{eq:spin_conservation}, and the associated Onsager coefficient \eqref{Onsager}. It follows that  for all interaction parameters $α\in \C$, $β,γ \in \R$, $a_i,b_i,c_i,d_i,e_i \in \C$ satisfying \eqref{condition} and $\sum_i(|a_i|^2 - |b_i|^2)\neq 0$, the Onsager coefficient is strictly positive $\mathfrak{L}\geq\mathfrak{L}_{\rm lower} > 0$. 
\end{theorem}

\begin{proof}
We proceed to calculate the bound \eqref{bound} for \eqref{hamiltonian_density}, \eqref{jump_operators}. Consider the connected correlation:
\begin{equation} \label{eq:numerator}
    \langle M, M,j^-\rangle^c = \langle M,M, j \rangle^c - χ^{-1}\langle M, M, s\rangle^c \langle M, j \rangle^c.
\end{equation}
As $\langle M,a\rangle^c = d\mathsf a/dμ$ and $\langle M,M,a\rangle^c = d^2\mathsf a/dμ^2$ for any observable $a$ with average $\mathsf a=\omega_μ(a)$, changing variable $μ\to \mathsf s=\omega_μ(s)$, the bound \eqref{bound} can then be written  in terms of the flux curvature $\mathsf j''(\mathsf s) = v'(\mathsf s)$ as
\begin{equation}\label{boundjprime}
    \mathfrak L_{\rm lower} = \frc{(\chi \mathsf j'')^2}{8v_{\rm LR}},\quad
    ' = \frac{d}{d\mathsf s}
\end{equation}
(this is valid in general in chaotic systems). Therefore, as
\begin{equation}\label{chiexplicit}
    χ = \frc1{(\cosh\mu)^{2}}>0,
\end{equation}
diffusion is strictly positive if and only if the flux curvature is non-zero (see also \cite[Eq 1]{doyon_2019_diffusion}). This above formula shows that if the hydrodynamic mode has state-dependent velocity $v(\mathsf s)$, then the above bound is strictly positive. With a state-dependent velocity, the hydrodynamic trajectory of the mode is affected by macroscopic fluctuations of the initial state, leading to diffusive effects. We interpret the bound above as arising from this diffusion mechanism.

The bound \eqref{bound} can be calculated by elementary means. First, note that by the general result of \cite{kobayashi_2022_vanishing_currents} the Hamiltonian part $j^H$ of the spin current \eqref{eq:spin_conservation}, vanishes in the Gibbs state $\omega_\mu(j^{\rm H}(x)) = 0$. Thus our bound vanishes in generic Hamiltonian spin chain, both for energy and spin: the lack of ballistic transport does not allow the above diffusion mechanism to apply. However the incoherent processes lead to non-vanishing current and non-zero diffusion strength.

Indeed, we calculate the flux curvature 
\begin{equation}\label{fluxcurvature}
    v'(\mathsf s) = \mathsf j''(\mathsf s) = \sum_i(|a_i|^2-|b_i|^2)
\end{equation}
resulting in a lower bound for the spin-spin Onsager coefficient:
\begin{equation}\label{lowerexplicit}
   \mathfrak L_{\rm lower} = \frac{  \big(  \sum_i(|a_i|^2-|b_i|^2)  \big)^2 }{ 32eζ(2) \tilde{V} \cosh^4(\mu)}    
\end{equation}
where $\t V$ is defined in \eqref{Vtilde}. We note that as soon as there is incoherent spin transport, $\sum_i|a_i|^2\neq \sum_i|b_i|^2$, spin is an interacting ballistic mode as according to \eqref{js} the hydrodynamic velocity $v = \mathsf j'$ depends on the state, and we have a strictly positive diffusion bound. This shows \Cref{theorem:1}.
\end{proof}
 
\section{Conclusion} 
We have shown that, in every translation-invariant chaotic nearest-neighbor open quantum spin-$1/2$ chain with strong spin conservation and local detailed balance, the spin-spin Onsager coefficient is strictly positive whenever there is incoherent spin transport (Eq.~\eqref{lowerexplicit}). A chaotic spin chain is one for which the space of extensive conserved charges is one-dimensional (spanned by the magnetisation), therefore for which spin is phenomenologically expected to satisfy a single-component hydrodynamic equation. We expect that for almost every set of interaction parameters the spin chain be chaotic, thus admitting positive diffusion (\Cref{theorem:1}). Rigorously establishing that chaoticity is generic remains a crucial open problem.

This is the first rigorous proof of strictly positive diffusion in chaotic spin chains, albeit for open chains, where the bound is for the difference of normal diffusion strength, and a diffusion strength associated to irreversibility of the dynamics. It is based on a general theorem of projection onto quadratically extensive charges, extending \cite{Prosen_2014_diffusion_bound} and obtained by combining results of \cite{doyon_2019_diffusion} and \cite{ampelogiannis_2024_clustering}. The bound encodes the mechanism by which diffusive effects arise because of fluctuations of hydrodynamic trajectories due to macroscopic, thermodynamic fluctuations in the initial state. The methods are easily extended to a finite number of explicit extensive charges, higher spins, and finite- or short-range interactions, and can also be applied to closed systems such as unitary quantum circuits and cellular automata, and we expect the results still hold. Importantly, the methods do not work, for instance, for energy or spin transport in generic Hamiltonian quantum chains, because in Gibbs states there can be no persistent (ballistic) energy or spin current \cite{kobayashi_2022_vanishing_currents}: hydrodynamic velocities vanish, forbidding the above mechanism.

According to the theory of fluctuating hydrodynamics \cite{spohn_2014_nonlinear}, the above mechanism should in fact lead to superdiffusion. Indeed, by expanding the Gibbs current around a spacetime stationary background as $\mathsf j(\mathsf s+\delta\mathsf s) \approx \mathsf j + \mathsf j' \delta \mathsf s + \mathsf j'' \delta \mathsf s^2/2$ and adding microscopic fluctuations, the hydrodynamic equation becomes $\p_t \delta \mathsf s + v\p_x\delta \mathsf s+ \mathsf j'' \p_x\delta \mathsf s^2 = \frac12 \mathfrak D \p_x^2\delta \mathsf s + \p_x \xi$ where $\xi$ is an appropriately normalized delta-correlated white noise in space and time. This gives rise to the KPZ universality class if $\mathsf j''\neq 0$, in which case $\langle s(x,t),s(0,0)\rangle^c$ expands {\em superdiffusively} ($\mathfrak L=\infty$). According to \eqref{boundjprime}, $\mathsf j''\neq 0$ is exactly the condition for our diffusion lower bound to be strictly positive. Thus, although we have shown strict positivity of the Onsager coefficient, in fact, on phenomenological grounds, it is expected to be infinite. Rigorously showing $\mathfrak L=\infty$ in generic ballistic spin chains is another crucial open problem; results of \cite{doyon_2019_diffusion} based on projections along with assumptions on the form of correlations may give a way forward.

\bmhead{Acknowledgment} We are grateful to Curt von Keyserlingk for pointing out a calculation mistake in an earlier version of this paper, and the general result of vanishing persistent currents of \cite{kobayashi_2022_vanishing_currents}. In the earlier version, a bound was attempted for Hamiltonian spin chains. But the bound obtained by our technique must be zero in this case, as the energy current vanishes by the aforementioned result. The work of BD has been supported by the Engineering and Physical Sciences Research Council (EPSRC) under grant number EP/W010194/1.

\begin{appendices}

\section{Proof of the Onsager Coefficient bound} \label{appendix4.1}
The proof of \eqref{bound} is done in \cite[Section 5.2]{doyon_2019_diffusion} and requires showing that $(M)_n \coloneqq (1-\mathbb{P}) \sum_{-n}^n σ^3_x σ^3_0$ (in the limit $n \to \infty$, as an element of $\mathcal{H}_1$) is be quadratically extensive conserved charge, i.e.\ it satisfies Points $1^{\prime},2^{\prime},3^{\prime}$ in \cite[Section 4.2]{doyon_2019_diffusion}. In the present set-up, we have two differences; we have a slightly different clustering property \eqref{eq:clustering} and  non-Hamiltonian evolution which does not act as an automorphism. We resolve these in the next two subsections.

\subsection{Time evolution}
Here we consider the proof of  Point $3^{\prime}$ \cite[Section 4.2]{doyon_2019_diffusion}: there exists $k>0$
\begin{equation} \label{eq:appendix_time_1}
    \lim_n \sup_{t_1,t_2 \in [0,kn]} | \langle (M)_n, A(t_1)-A(t_2) \rangle_{1}|=0
\end{equation}
for $A \in \mathfrak{U}_{\rm loc}$ and $(M)_n =  \sum_{-n}^n σ^3_x σ^3_0 $. This is shown in \cite{doyon_2019_diffusion} for Hamiltonian time evolution $τ_t:\mathfrak{U}  \xrightarrow{\sim} \mathfrak{U}$ and uses the  property $\langle A(t), B\rangle_1 = \langle A, B(t)\rangle_1 $ which holds by time invariance of the state $ω$ and the fact that $τ_t$ is an automorphism: $τ_t^*(AB)= τ_t(A)^* τ_t(B)^*$. Then, \eqref{eq:appendix_time_1} is written as an integral of a time derivative and time evolution is moved on $(M)_n$, where we can use the conservation law \cite[Eq. 121]{doyon_2019_diffusion}. However, Lindbladian  time evolution does not preserve the algebra, hence we are not generally able to move $τ_t^*$ around in $\langle A(t), B\rangle_1 $. However, we can take advantage of \eqref{adjointL} to write:
\begin{equation}
    \langle (M)_n , \mathcal{L}^*A \rangle_1 = \langle \mathcal{L}(M)_n , A \rangle_1
\end{equation}
For $\mathcal{L}(M)_n = \mathcal{L} \big( \sum_{-n}^n σ^3_x σ^3_0 \big)$ we can use an ``approximate Leibniz rule'' satisfied by $\mathcal{L}$ on observables supported on just one site:
\begin{equation}
    \mathcal{L}( q(x) q(y)) = \mathcal{L}(q(x)) q(y) + q(x)\mathcal{L}(q(y)) + e(x,y) , \quad q(x) \in \mathfrak{U}_{ \{x\}},q(y) \in \mathfrak{U}_{ \{y\}}
\end{equation}
with the ``error term" $e(x,y)\in \mathfrak{U}_{\rm loc}$ is supported on a finite region containing $x$ and $y$, vanishes for $|x-y|\geq 2$, and satisfies $\sum_{y=x-1}^{x+1} e(x,y) = \sum_{x=y-1}^{y+1} e(x,y)=0$ \footnote{This is  deduced directly from the fact that $\mathcal{L}$ is a sum over nearest-neighbor terms}. Explicitly, we find $e(x,x-1) = s(x)j(x-1)-j(x-1)s(x-1)$, $e(x,x) = (j(x)-j(x-1))s(x) + s(x)(j(x)-j(x-1))$, $e(x,x+1) = j(x)s(x+1)-s(x)j(x)$, and one can check that it satisfies the sum rules by using the relation $j(x)(s(x)+s(x+1)) = 0$ (recall that $s(x) = \sigma_x^3$). Then, we have
\begin{equation}
   \sum_{x=-n}^n\langle  \mathcal{L}(σ^3_x) σ^3_0 + σ^3_x\mathcal{L}(σ^3_0), A \rangle_1
\end{equation}
and we can use the conservation laws $\mathcal{L}(σ^3_x) = {j}(x) - {j}(x-1)$ (that is, the current with respect to backward time evolution is $-j$), use space invariance of $\langle \cdot,\cdot \rangle_1$ and perfom the sum over $x$ to get as in \cite[Eq. 121]{doyon_2019_diffusion}:
\begin{equation}
     \langle \mathcal{L}(M)_n, A \rangle_1 \big|_{t=0} =  -\langle (M)_n, \mathcal{L}^*A \rangle_1 \big|_{t=0}= \langle  {j}(-n) σ^3_0 -  {j}(n) σ^3_0  -σ^3_{-n} {j}_0 +  σ^3_{n} {j}_0  , A \rangle_1.
\end{equation}
The rest of the proof follows from \cite[Eq. 122]{doyon_2019_diffusion}, where we are left with connected correlations of the form $\int_{t_1}^{t_2} \sum_{x\in \Z} \langle  j(n,0) σ^3(0,0) , A(x,t) \rangle^c  dt$, that are discussed below and vanish by clustering.

\subsection{Clustering} \label{appendix4}
The proof of the Onsager matrix bound in \cite[Theorem 5.1]{doyon_2019_diffusion} is done under the assumption that for any $A_1,\ldots,A_n \in \mathfrak{U}_{\rm loc}$, any $x_1,\ldots, x_n \in \Z$, $t_1,\ldots , t_n \in \R$ there exists $b,c>0$ so that
\begin{equation}
    \langle A_1(x_1,t_1),\ldots , A_n(x_n,t_n) \rangle^c \leq f( \max_i \min_j \{ w_{ij} \})
\end{equation}
where $w_{ij}=|x_i-x_j| - υ_{\rm LR} |t_i-t_j|$, $f(z)= c(b+z)^{-d}$. This is slightly different from the rigorously established clustering property \eqref{eq:clustering} shown in \cite{ampelogiannis_2024_clustering}, which instead holds for times $|t_i| \leq υ^{-1}z$, $z=\max\min \{\dist(A_i(x_i),A_j(x_j)) \}$ and bounds the connected correlations by a function of the max-min spatial distances and not the max-min space-time distances. However, the proof remains essentially unchanged.  To see this, first note that \eqref{eq:clustering} can be expressed in terms of the distances $|x_i-x_j|$ instead of  $\dist(A_i(x_i),A_j(x_j))$ at the cost of a larger constant $C_{A_1,\ldots, A_n}$ that takes into account the diameter of $\supp(A_i)\cup\supp(A_j)$, see for example \cite[Eq 121, Appendix C]{ampelogiannis_2023_almost}. Note also that we only need to consider the last part of the proof of \cite[Theorem 5.1]{doyon_2019_diffusion}, Equation (123), that requires use of clustering of connected correlators with observables at different times, as the rest of the proof has equal-time correlators and remains the same. Consider \cite[Eq. 123]{doyon_2019_diffusion}: let $A,B,C \in \mathfrak{U}_{\rm loc}$ and the term $I_{s_1,s_2}(n)\coloneqq\int_{s_1}^{s_2} dt \langle A(0,t),B(n,0)C\rangle_1$. We need to show that this vanishes when $s_1,s_2 \leq n/(2υ_{\rm LR})$ as $n \to \infty$. It is sufficient to consider $s_1=0$, $s_2= nk$ for $k< 1/(2υ_{\rm LR}) $. We have 
\begin{equation}
  \begin{array}{*3{>{\displaystyle}l}p{5cm}}
    |I_{s_1,s_2}(n)| \leq \int_{s_1}^{s_2} dt \sum_{x \in \Z} | \langle A(x,t),B(n,0)C(0,0)\rangle^c| \\
     = \int_{s_1}^{s_2} dt \int_{x \in \R} dx| \langle A(\lfloor x \rfloor,t),B(n,0) C(0,0)\rangle^c| \\
     = n^2 \int_{0}^{k} dt \int_{x \in \R} dx| \langle A(\lfloor nx \rfloor,nt),B(n,0) C(0,0)\rangle^c|
    \end{array}
\end{equation}  

Following \cite[Eq 124]{doyon_2019_diffusion}, we  split the space integral into two regions $R_1 = \{ |x - \frac12| > \frac32+kυ_{{LR}} \}$ and $R_2 = \{  |x - \frac12| \leq \frac32+kυ_{{LR}} \}$, with $t \leq k$. By  \cite[Lemma V.5]{ampelogiannis_2024_clustering} clustering is controlled by $z = \min( \lfloor nx \rfloor-n, \lfloor nx \rfloor)$ if $t\leq z/υ_{LR}$, which holds in both $R_1,R_2$. Thus, the integrand is bounded exponentially and $I_{s_1,s_2}(n)$ vanishes at $n \to \infty$

\section{Reversible Lindbladian spin transport}\label{apprev}

We now show that, under the conditions
\begin{equation}\label{conditionreversibility}
    \alpha = 0,\quad \sum_i(c_i-d_i)=0,\quad \sum_i(\b e_i a_i - e_i\b b_i)=0,
\end{equation}
which in particular imply local detailed balance \eqref{detailed_balance}, spin transport with Lindbladian $\mathcal L^*$ is reversible on $\mathcal H_k$, $k=0,1$:
\begin{equation}
    \tau^*_t(s(x)) = \tau_{-t}(s(x))
\end{equation}
for all $t\in\R,\,x\in\Z$. In fact, this holds on the subspace $\mathcal H_{k,\rm spin}$ of $\mathcal H_k$ obtained from (the closure of) the space of local operators
\begin{equation}
    \mathfrak U_{\rm spin} = {\rm span}_\Lambda\Bigg(\prod_{y\in\Lambda}\sigma_y^3\Bigg)\ni s(x),
\end{equation}
that is $\tau^*_t(A) = \tau_{-t}(A)$ for all $A\in\mathcal H_{k, \rm spin}$, and $\mathcal L^* = -\mathcal L$ on $\mathfrak U_{\rm spin}$. In particular, both $\tau$ and $\tau^*$ form {\em unitary groups} on $\mathcal H_{k,\rm spin}$ (which we believe can be shown to be strongly continuous). Hence, the setup of \cite[Secs 4, 5]{doyon_hydrodynamic_2022} applies and projection theorems hold. Note that on $\mathfrak U_{\rm spin}$, the Hamiltonian $H$ with $\alpha=0$ acts trivially -- thus conditions \eqref{conditionreversibility} and the restriction to these observables gives a model of pure quantum jumps.

For the proof, we only need to show that $\mathcal L^* = -\mathcal L$ on $\mathfrak U_{\rm spin}$. We note that $\mathcal L^* = \sum_{x\in\Z} \mathcal L^*(x)$ where $\mathcal L^*(x)(A) = i[h(x),A] + \frc12\sum_i L_i(x)^\dag[A,L_i(x)] + [L_i(x)^\dag,A]L_i(x)$, and similarly for $\mathcal L$. Further, we have
\begin{align}
    &\mathcal L^*(x)(\1_x\1_{x+1}) = 0,\ 
    \mathcal L^*(x)(\sigma^3_x \1_{x+1}) = -j(x),\\
    &\mathcal L^*(x)(\1_x\sigma^3_{x+1}) = j(x),\ 
    \mathcal L^*(x)(\sigma^3_x\sigma^3_{x+1}) = 0.
\end{align}
Under the condition \eqref{conditionreversibility} we see that (1) the current $j(x)=j^{\rm L}(x)$, Eq.~\eqref{Lcurrent} (recall that $j^{\rm H}(x)=0$), lies in $\mathfrak U_{\rm spin}$, and (2) $\mathcal L$, obtained from exchanging $b_i\leftrightarrow \b a_i$ and $c_i,d_i,e_i\leftrightarrow \b c_i, \b d_i,\b e_i$, gives the same results with $j(x)\leftrightarrow -j(x)$. Therefore $\mathcal L^*(x)(o_xo_{x+1})=-\mathcal L(x)(o_xo_{x+1})\in\mathfrak U_{\rm spin}$ for $o_xo_{x+1}\in\mathfrak U_{\rm spin}$. Thus on any operator $A=\prod_{y\in\Z} o_y\in \mathfrak U_{\rm spin}$ with $o_x\in\{\1_x,\sigma^3_x\}$ we have $\mathcal L^*(A) = \sum_{x\in\Z}\Big(\prod_{y\in\Z\setminus\{x,x+1\}}o_y\Big)\mathcal L^*(o_xo_{x+1}) = -\mathcal L(A)\in \mathfrak U_{\rm spin}$, which completes the proof.


\end{appendices}

\bibliography{references}

\end{document}